\documentclass{article}

\usepackage{arxiv}

\usepackage[utf8]{inputenc} 
\usepackage[T1]{fontenc}    
\usepackage{hyperref}       
\usepackage{url}            
\usepackage{booktabs}       
\usepackage{amsfonts}       
\usepackage{nicefrac}       
\usepackage{microtype}      
\usepackage{lipsum}
\usepackage{graphicx}
\graphicspath{ {./images/} }
\usepackage{amssymb}
\usepackage{amsmath}
\usepackage{lineno,hyperref}
\usepackage{latexsym}
\usepackage{amsthm}
\usepackage{amssymb}
\usepackage{amsmath,float,cancel}
\usepackage{ulem}
\usepackage{tikz}
\usepackage{tkz-graph}
\usetikzlibrary{decorations,arrows,shapes}
\usepackage{enumerate}
\usepackage{comment}

\newtheorem{teo}{Theorem}[section]
\newtheorem{lem}[teo]{Lemma}

\newtheorem{obs}[teo]{Remark}
\newtheorem{prop}[teo]{Proposition}

\newtheorem{cor}[teo]{Corollary}
{\theoremstyle{definition}}

\title{On the distribution of $A_{\alpha}$-eigenvalues in terms of graph invariants}

\author{
 Uilton Cesar Peres Junior \\
  Centro Federal de Educação Tecnológica Celso Suckow da Fonseca - CEFET/RJ\\
  Rio de Janeiro, RJ \\
  \texttt{uilton.junior@aluno.cefet-rj.br} \\
   \And
 Carla Silva Oliveira \\
  Escola Nacional de Ciências Estatísticas - ENCE/IBGE\\
  Rio de Janeiro, RJ \\
  \texttt{carla.oliveira@ibge.gov.br} \\
  \And
 André Ebling Brondani \\
  Universidade Federal Fluminense - UFF\\
  Volta Redonda, RJ\\
  \texttt{andrebrondani@id.uff.br} \\
}

\begin{document}
\maketitle
\begin{abstract}
Let $G$ be a connected graph of order $n$, and $A(G)$ and $D(G)$ its adjacency and degree diagonal matrices, respectively. For a parameter $\alpha \in [0,1]$, \cite{Nikiforov2017} introduced the convex combination $A_{\alpha}(G) = \alpha D(G) + (1 - \alpha)A(G)$. In this paper, we investigate the spectral distribution of $A_\alpha(G)$-eigenvalues, over subintervals of the real line. We establish lower and upper bounds on the number of such eigenvalues in terms of structural parameters of $G$, including the number of pendant and quasi-pendant vertices, the domination number, the matching number, and the edge covering number. Additionally, we exhibit families of graphs for which these bounds are attained.
\end{abstract}


\section{Introduction}

Historically, the Spectral Graph Theory has played an essential role in mathematical chemistry. In particular, H{\"u}ckel Molecular Orbital (HMO) theory models the $\pi$-electron systems of conjugated hydrocarbons, where the spectrum of the molecular graph correlates with the compound's reactivity and stability~\cite{Huckel1931, Gutman2011}. Within this context, the nullity \( \eta(G) \), defined as the multiplicity of the zero eigenvalue of \( A(G) \), is interpreted as an indicator of chemical instability.

Several studies have explored the distribution of $L(G)$-eigenvalues based on specific graph parameters. In \cite{GroneMerrisSunder1990}, the authors show lower bounds for the distribution of $L(G)$-eigenvalues in the intervals $[0,1]$ and $[1, \infty)$ in terms of the number of pendant vertices. In \cite{GuoWuZhang2011}, the authors presented results for the distribution of $L(G)$-eigenvalues in terms of edge covering number and number of quasi-pendant vertices. Furthermore, in \cite{HedetniemiJacobsTrevisan2016}, the authors presented a result about the distribution of $L(G)$-eigenvalues in terms of the order of the graph. For further studies establishing bounds for the number of eigenvalues of Laplacian and signless Laplacian matrices in specific intervals based on other parameters, such as domination number and matching number, see \cite{CardosoJacobsTrevisan2017, GuoXue2024, Merris1991, MingWang2001, FanGongZhouTanWang2007, ZhouZhouDu2015, JacobsOliveiraTrevisan2021, AhanjidehAkbariFakharanTrevisan2022, XuZhou2024-2}.

\cite{Nikiforov2017} introduced the set of $A_{\alpha}$-matrices associated with a graph $G$ as a convex combination of the adjacency matrix, \( A(G) \), and the diagonal degree matrix, \( D(G) \) defined the following way: \( A_\alpha(G) = \alpha D(G) + (1 - \alpha)A(G) \), where \( \alpha \in [0,1]\). Through this set of matrices and motivated by the above mentioned works about the eigenvalue distribution, our main objective in this paper is to obtain bounds for the distribution of $A_{\alpha}$-eigenvalues in terms of graph invariants. In some cases, we identify graph families for which the bounds are attained, and we obtain results for the adjacency and signless Laplacian matrices. 

The paper is organized as follows. Section~\ref{sec:pre} presents preliminary results about Graph Theory and Spectral Graph Theory. In Section~\ref{sec:main1}, the results of the distribution of $A_{\alpha}$-eigenvalues in intervals involving the number of pendant and quasi-pendant vertices are presented. Section~\ref{sec:main2} explores additional bounds for the distribution of $A_{\alpha}$-eigenvalues in intervals based on the domination number, matching number, edge covering number, and other invariants of graphs.

\section{Preliminaries}\label{sec:pre}

This section presents initial concepts and results from the literature to obtain the main contributions of this paper.

Let $G=(V(G), E(G))$ be a simple graph of order $n$ and size $m$. When no ambiguity arises, we simply write $G=(V, E)$. The \textit{induced subgraph} $G[S]$ is the graph formed by the vertex set $S \subseteq V(G)$ and all edges of $G$ with both endpoints in $S$. The degree of a vertex $v\in V(G)$ is denoted by $d_G(v)$ or simply $d(v)$. The \textit{degree sequence} of $G$ is the non-increasing sequence of the degrees of the vertices, given by $d(G) = (d(v_1), d(v_2), \dots, d(v_n))$. The minimum and maximum degrees of $G$ are denoted by $\delta(G)=d(v_n)$ and $\Delta(G)=d(v_1)$, respectively. When there is no doubt, we use $\delta$ and $\Delta$. A vertex is \textit{isolated} if $d(v) = 0$, pendant if $d(v) = 1$, and quasi-pendant if it is adjacent to a pendant vertex. The sets $\mathcal{P}(G)$ and $\mathcal{Q}(G)$ consist of the pendant and quasi-pendant vertices of $G$, respectively. A vertex $v$ is called \textit{internal} if $d(v) \ge 2$. The neighborhood of a vertex $v$, denoted $N_G(v)$, is the set of all vertices adjacent to $v$ and $|N_G(v)|=d(v)$. An edge of $G$ is called a \textit{pendant edge} if it is incident with a pendant vertex.

A \textit{path} in $G$ is a sequence of distinct vertices $(v_1,\ldots,v_r)$ such that $\{v_i, v_{i+1}\} \in E(G)$ for each $i \in \{1,\ldots,r-1\}$. A \textit{cycle} is a sequence $(v_1,\ldots,v_r,v_1)$, where $r\ge 3$, $v_1,\ldots,v_r$ are distinct, $\{v_r,v_1\} \in E(G)$, and $\{v_i, v_{i+1}\}\in E(G)$ for each $i \in \{1,\ldots,r-1\}$.
A \textit{Hamiltonian path} and a \textit{Hamiltonian cycle} are a path and a cycle, respectively, that contains all the vertices of a graph $G$. A graph $G$ is said to be \textit{connected} if there is a path between every pair of its vertices. A connected component of $G$ is a maximal set of vertices $P \subset V(G)$ such that for each $u \in P$ and $v \in P$, there exists a path in $G$ from vertex $u$ to vertex $v$. The \textit{complement} of a graph $G=(V(G),E(G))$ is the graph $\overline{G}=(V(G), E(\overline{G})$ where $\{u, v\} \in E(\overline{G})$ if and only if $\{u, v\} \notin E(G)$.

A graph is \text{bipartite} if its vertex set can be partitioned into disjoint sets, $V_1$ and $V_2$, such that all edges join vertices from different parts. The complete bipartite graph is denoted by $K_{p,q}$. A graph is called \textit{$r$-regular}  if each of its vertices has degree $r$. The \textit{complete graph, path, cycle} and \textit{star}  of order $n$ are denoted by $K_{n}$, $P_{n}$, $C_{n}$ and $S_{n}$, respectively. A \textit{double star} is a tree obtained from two vertex-disjoint stars by adding an edge between their centers. A \textit{forest} is a graph without cycles, and a \textit{tree} is a connected forest. A \textit{spanning forest} of $G$ is a forest whose vertex set coincides with the set of vertices of $G$.

An \textit{independent set} of $G$ is a subset $S \subseteq V(G)$ such that no two vertices in $S$ are adjacent. A \textit{dominating set} in $G$ is a subset $C$ of its vertices, such that any vertex of $G$ is in $C$, or has a neighbor in $C$. The domination number, $\gamma(G)$, is the number of vertices in a smallest dominating set for $G$. A \textit{matching} in $G$ is a set of edges no two of which share a common vertex. A \textit{maximum matching} is a matching that contains the largest possible number of edges. The matching number $\nu (G)$ of a graph $G$ is the size of a maximum matching. An \textit{edge cover} is a set of edges such that every vertex is incident to at least one edge in the set, and the minimum size of such a set is the \textit{edge covering number},  $\beta(G)$.

The next result establishes a relation between the order, the matching number and the edge covering number of a graph.

\begin{lem}[\cite{Galli1959}]\label{Galli1959}
If $G$ is a graph of order $n$ with no isolated vertices, then $\beta(G) + \nu(G) = n$.
\end{lem}

The next lemma establishes a relation between $\gamma(G)$ and the existence of a spanning forest of $G$ composed of $\gamma(G)$ stars.

\begin{lem}[\cite{CardosoJacobsTrevisan2017}]\label{domtrelem6}
Every graph $G = (V, E)$ that contains no isolated vertices, with domination number $\gamma(G)$, has a star forest $F = S_{n_1} \cup \dots \cup S_{n_{\gamma(G)}}$ as its spanning forest, such that each $v \in V$ belongs to exactly one star, and the star centers form a minimal dominating set of $G$.
\end{lem}

Let $G = (V_1, E_1)$ and $H = (V_2,E_2)$ be two graphs of orders $n_1$ and $n_2$, respectively, and let $e$ be an edge in $G$. The \textit{union} of $G$ and $H$ is the graph $G \cup H = (V_1 \cup V_2, E_1 \cup E_2)$. The graph $kG$, where $k$ is a positive integer, is the graph obtained from the union of $k$ copies of the graph $G$. Let $G$ and $H$ be graphs, with $G$ of order $n$. According to \cite{FruchtHarary1970}, the corona of $G$ and $H$, denoted by $G \circ H$, is the graph obtained from one copy of $G$ and $n$ disjoint copies of $H$ by joining each vertex of $G$ to every vertex in its corresponding copy of $H$.

Let $M$ be a square matrix of order $n$. The $M$-characteristic polynomial is defined by $p_M(x) =  det (x I_n - M) $
and the roots of $p_M(x)$ are called the $M$-eigenvalues, where $I_n$ is the identity matrix of order $n$. For $M$ symmetric, its eigenvalues $\lambda_j(M)$, $1 \leq j \leq n$, are arranged in non-increasing order, that is, $  \lambda_{1}(M) \geq \ldots \geq \lambda_{n-1}(M) \geq \lambda_{n}(M).$ The \textit{M-spectrum} is the multiset of its eigenvalues, denoted by $\sigma(M)=\{\lambda_{1}(M), \lambda_{2}(M), \dots, \lambda_{n}(M)\}$.  We denote by $ m_{M}(I)$ the number of eigenvalues of $M$  in the  interval $I$ including its multiplicities and write $m_M(I)=m_M(a)$, when $I=[a,a]$. In particular, the eigenvalues of a matrix $M$ associated with a graph $G$ are referred to as the $M(G)$-eigenvalues, and the $M(G)$-characteristic polynomial is denoted by $p_{M}(G,x).$ We denote by $\mathbf{0}_n$ the zero matrix of order $n$.

The adjacency matrix of a graph $G$ is defined by $A(G) = [a_{ij}]$, where $a_{ij} = 1$ if $\{v_i, v_j\} \in E$ and $a_{ij} = 0$, otherwise. The diagonal matrix of degrees of $G$, $D(G) = [d_{ij}]$, is defined by $d_{ii} = d(v_i)$, and $d_{ij} = 0$, $\forall i \neq j$, the Laplacian matrix of $G$, $L(G)$, is defined as $L(G) = D(G) - A(G)$ and the signless Laplacian, denoted by $Q(G)$, is defined as $Q(G) = D(G) + A(G)$. 

\cite{Nikiforov2017} defined the set of matrices $A_{\alpha}(G)$ as a convex linear combination of $A(G)$ and $D(G)$ in the following way
$$A_{\alpha}(G) = \alpha D(G) + (1-\alpha)A(G), \mbox{ for } \alpha \in [0,1].$$

This set of matrices contains some of the classical matrices associated with a graph:  \( A_0(G) = A(G) \), \( A_1(G) = D(G) \), and \( A_{1/2}(G) = \frac{1}{2}Q(G) \), where $Q(G)$ is the signless Laplacian matrix of $G$.
Although \( A_{\alpha}(G) \) does not contain the Laplacian matrix, the following relationship holds: \( A_\alpha(G) - A_\beta(G) = (\alpha - \beta)L(G) \) for all $\alpha, \beta \in [0,1].$

The following interlacing result can be found in \cite{Haemers1995}.

\begin{teo}[\cite{Haemers1995}]\label{interlacing:matrizes}
Let $M$ be a symmetric matrix of order $n$, and let $B$ be a principal submatrix of $M$ of order $r<n$. Then, for all $i=1,2, \dots, r$,
\begin{equation*}
 \lambda_{i}(M) \geq \lambda_{i}(B) \geq \lambda_{n-r+i}(M).   
\end{equation*}
\end{teo}

The next results involve bounds for the $A_{\alpha}(G)$-eigenvalues.

\begin{lem}[\cite{Nikiforov2017}]\label{cotaniki}
Let $G$ be a graph of order $n$. If $\alpha \in [0,1]$ and $1 \leq k \leq n$, then
  $$\lambda_k(A_{\alpha}(G)) \leq d(v_k).$$
  In particular, $\lambda_1(A_{\alpha}(G)) \leq \Delta(G).$
\end{lem}

\begin{lem}[\cite{LiuDasShu2020}]\label{cotakinkar}
     Let $G$ be a graph of order $n$. If $\alpha \in [0,1]$ and $1 \leq k \leq n$, then
     $$\lambda_k(A_{\alpha}(G)) \geq \alpha d(v_k) - (1-\alpha) \sqrt{\left\lceil \dfrac{k}{2} \right\rceil \left\lfloor\dfrac{k}{2} \right \rfloor}.$$
     If the equality holds and $\alpha \in (0,1)$, then $G$ have an induced subgraph $H \cong K_{\left\lceil\frac{k}{2}\right\rceil, \left\lfloor\frac{k}{2}\right\rfloor}$ such that $d_G(v_i) = \Delta(G)$ for all $v_i \in V(H).$
\end{lem}

As a consequence of Lemmas \ref{cotaniki} and \ref{cotakinkar}, we have the following corollary.

\begin{cor}\label{cotasalphaautovalores}
Let $G$ be a graph of order $n$. If $\alpha \in [0,1]$ and $k\in \{1,2,\dots,n\}$, then
\begin{equation*}
\lambda_k(A_{\alpha}(G)) \in \left[\alpha \delta(G) -(1-\alpha)\sqrt{\left\lceil \dfrac{n}{2} \right\rceil \left\lfloor\dfrac{n}{2} \right \rfloor} , \Delta(G)\right].    
\end{equation*}
\end{cor}

The following proposition presents the $A_{\alpha}(G)$-spectrum when $G$ is the star or the complete graph.

\vspace{0.7cm}

\begin{prop}[\cite{Nikiforov2017}]\label{autokn}\label{starspectrum}
    If $\alpha \in [0,1]$, then
    \begin{enumerate}[(i)]
        \item $\lambda_k(A_{\alpha}(S_n)) = \begin{cases}
			\dfrac{1}{2}\left(\alpha n + \sqrt{\alpha^2n^2 + 4(n-1)(1-2\alpha)} \right), & \text{if } k = 1;\\[10pt]
			\dfrac{1}{2}\left(\alpha n - \sqrt{\alpha^2n^2 + 4(n-1)(1-2\alpha)} \right), & \text{if } k = n;\\
            \alpha, &  \text{if } 2 \leq k \leq n-1;
    \end{cases}$
    \item $\lambda_k(A_{\alpha}(K_n)) = \begin{cases}
			n-1, & \text{if } k = 1;\\[10pt]
            \alpha n -1, &  \text{if } 2 \leq k \leq n;
    \end{cases}$
    \end{enumerate}
\end{prop}

Let $A, B$ be real symmetric matrices of order $n$. We say that $B$ is \textit{congruent} to $A$ if there exists a nonsingular matrix $P$ such that $B = P^T A P$. The following result relates the number of positive and negative eigenvalues of two congruent matrices.

\begin{teo}[\cite{Horn2013}]\label{sylvesterlaw}
Let $A,B$ be real symmetric matrices of order $n$. The matrices $A$ and $B$ are congruent if and only if they have the same number of positive eigenvalues and the same number of negative eigenvalues.
\end{teo}

Let $N_{\alpha}(G)$ be the submatrix of $A_{\alpha}(G)$ corresponding to internal non-quasipendant vertices. The number of pendant vertices and quasi-pendant vertices of a graph $G$ is denoted by $p(G)$ and $q(G),$ respectively. 
Theorem~\ref{teo:germaincardoso} presents the multiplicity of $\alpha$ as a $A_{\alpha}(G)$-eigenvalue.

\begin{teo}[\cite{PastenCardosoRojo2018}]\label{teo:germaincardoso}
    If $G$ is a graph and $\alpha \in [0,1)$ then
    $$m_{A_{\alpha}(G)}(\alpha) = p(G) - q(G) + m_{N_{\alpha}(G)}(\alpha).$$
\end{teo}

Proposition~\ref{prop:linxue} establishes an interlacing relation between the eigenvalues of the matrices $A_{\alpha}(G)$ and $A_{\alpha}(G-e)$, where $G-e$ is the subgraph of $G$ obtained by removing the edge $e$.

\begin{prop}[\cite{LinXueShu2018}]\label{prop:linxue}
    Let $G$ be a graph of order $n$. If $e \in E(G)$ and $\alpha \in [1/2, 1]$, then
    $$\lambda_i(A_{\alpha}(G)) \geq \lambda_i(A_{\alpha}(G-e))$$
    for $1 \leq i \leq n$.
\end{prop}

The condition $\alpha \ge 1/2$ in Proposition~\ref{prop:linxue} is sharp. For any edge $e \in E(G)$, the only non-zero eigenvalues of the difference matrix $A_{\alpha}(G) - A_{\alpha}(G-e)$ are $1$ and $2\alpha-1$. Consequently, this difference matrix is positive semidefinite if and only if $\alpha \ge 1/2$. For $\alpha < 1/2$, the difference matrix is indefinite, and the monotonicity of the eigenvalues under edge removal no longer holds in general.

Proposition \ref{prop:positive} shows when $A_{\alpha}(G)$ is positive semidefinite or positive definite.

\begin{prop}[\cite{NikiforovRojo2017}]\label{prop:positive}
 If $\alpha \in [\frac{1}{2},1]$, then $A_{\alpha}(G)$ is positive semidefinite. If $\alpha\in (\frac{1}{2},1]$ and $G$ has no isolated vertices, then $A_{\alpha}(G)$ is positive definite.   
\end{prop}

\section{Distribution of $A_{\alpha}(G)$-eigenvalues involving $p(G)$ and $q(G)$}\label{sec:main1}

In this section, we present lower bounds for the distribution of $A_{\alpha}(G)$-eigenvalues in intervals involving $p(G)$ and $q(G)$.
For this purpose, let $G$ be a graph with $\delta(G)=~1$, equivalently, $p(G) \geq 1.$ From Corollary ~\ref{cotasalphaautovalores}, we have that $m_{A_{\alpha}}[\psi, \Delta(G)]=n,$ where $\psi=\alpha - (1 - \alpha)\sqrt{\lceil n/2 \rceil \lfloor n/2 \rfloor}.$ The next result presents the distribution of $A_{\alpha}(G)$-eigenvalues in the intervals $[\psi, \alpha]$ and $[\alpha, \Delta(G)]$.

\begin{teo}\label{org:dist1}
Let $G$ be a connected graph of order $n \geq 3$ and $\alpha \in [0,1]$. Then
$$m_{A_{\alpha}(G)}[\psi, \alpha] \geq p(G) \quad \text{and} \quad m_{A_{\alpha}(G)}[\alpha, \Delta(G)] \geq p(G).$$
Moreover, for $\alpha\in[0,1)$, both equalities hold for the star $S_n$ and for the corona $H\circ K_1$, where $H$ is connected.
\end{teo}
\begin{proof}
If $p(G) = 0$, the result holds trivially. Now, suppose $p(G)=p>0.$ Labeling the pendant vertices first; then, by deleting the last $n-p$ rows and columns of $A_{\alpha}(G)$, we obtain the principal submatrix $\alpha I_p$.
    By Theorem~\ref{interlacing:matrizes}, it follows that
\begin{equation*}
\lambda_i(A_{\alpha}(G)) \geq \alpha \geq \lambda_{i+(n-p)}(A_{\alpha}(G)), \quad 1 \leq i \leq p.   
\end{equation*}

Hence, at least $p$ eigenvalues of $A_{\alpha}(G)$ are not smaller than $\alpha$, and at least $p$ are not greater than $\alpha$. From Corollary~\ref{cotasalphaautovalores}, the inequalities follow.

\medskip
For the star $S_n$, since $p(S_n)=n-1$, Proposition~\ref{starspectrum} gives 
\begin{equation*}
m_{A_{\alpha}(S_n)}[\psi, \alpha] 
= m_{A_{\alpha}(S_n)}[\alpha, \Delta(S_n)] 
= p(S_n).    
\end{equation*}

Now, let $H$ be a connected graph of order $n_1$, and let $G \cong H \circ K_1$. Then $G$ has $2n_1$ vertices with $n_1$ pendants. Labeling the pendant vertices first, we obtain
$$
A_{\alpha}(G) = 
\begin{bmatrix}
    \alpha I_{n_1} & (1-\alpha)I_{n_1} \\
    (1-\alpha)I_{n_1} & A_{\alpha}(H) + \alpha I_{n_1}
\end{bmatrix}.
$$

From \cite[Theorem~3]{Silvester2000},
\[
p_{A_{\alpha}}(G, x) 
= \det\!\left([(x-\alpha)^2-(1-\alpha)^2]I_{n_1}-(x-\alpha)A_{\alpha}(H)\right).
\]
Therefore,
\begin{equation*}
p_{A_{\alpha}}(G, x) = \prod_{i=1}^{n_1}\left(x^2-(2\alpha+\lambda_i(A_{\alpha}(H)))x + \alpha(\lambda_i(A_{\alpha}(H) +2)-1\right),   
\end{equation*}
and consequently, the two $A_{\alpha}(G)$-eigenvalues corresponding to $\lambda_i(A_{\alpha}(H))$ are
\[
\dfrac{2\alpha + \lambda_i(A_{\alpha}(H)) \pm \sqrt{\omega}}{2},
\]
where $\omega = 4(\alpha-1)^2+\lambda_i^2(A_{\alpha}(H))$, for $1 \leq i \leq n_1$.

\medskip

Since, for $\alpha \in [0,1]$, $-\sqrt{\omega} \le \lambda_i(A_{\alpha}(H)) \le \sqrt{\omega}$, 
the left inequality gives
\begin{equation*}
\alpha \le \frac{2\alpha + \lambda_i(A_{\alpha}(H)) + \sqrt{\omega}}{2},    
\end{equation*}
while the right inequality gives
\begin{equation*}
\frac{2\alpha + \lambda_i(A_{\alpha}(H)) - \sqrt{\omega}}{2} \le \alpha.    
\end{equation*}

Therefore, for each eigenvalue $\lambda_i(A_{\alpha}(H))$, one eigenvalue of $A_{\alpha}(G)$ is greater than $\alpha$ and another is  smaller than $\alpha$. Consequently,
\[
m_{A_{\alpha}(G)}[\psi,\alpha] 
= m_{A_{\alpha}(G)}[\alpha,\Delta(G)] 
= n_1 = p(G),
\]
showing that the bounds are attained for $G \cong H \circ K_1$.

\end{proof}

It is worth noting that if $G \cong \bigcup_{i=1}^k G_i$ for $k \ge 1$, where each $G_i$ is a connected graph of order $n \ge 3$ with $p(G_i) = p_i > 0$, then $\sigma(A_{\alpha}(G)) = \bigcup_{i=1}^k \sigma(A_{\alpha}(G_i))$, and the following corollary is obtained.

\begin{cor}\label{distuniao}
    Let $G = \bigcup_{i=1}^k G_i$ for $k \ge 1$, where $G_i$ is a connected graph of order $n_i\ge 3$ for $1 \le i \le k$ with $p_i = p(G_i) > 0$. If $\alpha \in [0,1]$, then
    $$m_{A_{\alpha}(G)}[\psi, \alpha] \ge p_1 + p_2 + \cdots + p_k \quad \text{and} \quad m_{A_{\alpha}(G)}[\alpha, \Delta(G)] \ge p_1 + p_2 + \cdots + p_k.$$
\end{cor}

As a consequence of Theorem~\ref{org:dist1}, we obtain the distribution of $A(G)$-eigenvalues in the intervals $\left[-\sqrt{\left\lceil \frac{n}{2} \right\rceil \left\lfloor\frac{n}{2} \right \rfloor},0\right]$ and $[0, \Delta(G)]$ and for $Q(G)$-eigenvalues in the intervals $[0,1]$ and $[1, 2\Delta(G)]$ in terms of $p(G)$. These results are presented in Corollary~\ref{orgobs:dist1}.

\begin{cor}\label{orgobs:dist1}
    If $G$ is a connected graph of order $n \ge 3$ with $p(G)>0$, then
    \begin{enumerate}[(i)]
        \item $m_{A(G)}\bigl[-\sqrt{\left\lceil \frac{n}{2} \right\rceil \left\lfloor\frac{n}{2} \right \rfloor},0\bigr] \ge p(G)$;
        \item $m_{A(G)}[0, \Delta(G)] \ge p(G)$;
        \item $m_{Q(G)}[0, 1] \ge p(G)$;
        \item $m_{Q(G)}[1, 2\Delta(G)] \ge p(G)$.
    \end{enumerate}
\end{cor}

The following theorem provides a lower bound on the number of $A_{\alpha}(G)$-eigenvalues  lying in the interval $[\psi, \alpha)$, expressed in terms of the number of quasi-pendant vertices.

\begin{teo}\label{cotaquasi}
Let $G$ be a connected graph of order $n \ge 3$. If $\alpha \in [0, 1)$, then
$$
m_{A_\alpha(G)}[\psi, \alpha) \ge q(G)
\,\, \text{ and } \,\,
m_{A_\alpha(G)}(\alpha, \Delta(G)] \ge q(G).
$$
\end{teo}
\begin{proof}
If $q(G) = 0$, the result holds trivially. Now, suppose $q(G) = q > 0$ and $v_1, \dots, v_q$ be the quasi-pendant vertices of $G$. For each $i$, choose a pendant vertex $u_i$ adjacent to $v_i$. As $G$ is connected and has order at least $3$, the vertices
$$
u_1, \dots, u_q, v_1, \dots, v_q
$$
are pairwise distinct. Labelling the pendant vertices and then the quasi-pendents vertices of $G$, we can obtain the  principal submatrix $B$ of $A_\alpha(G) - \alpha I$ as follows 
$$
B = 
\begin{pmatrix}
\mathbf{0}_q & (1-\alpha)I_q \\
(1-\alpha)I_q & C
\end{pmatrix},
$$
where $C$ is a real symmetric matrix. Now, defining the nonsingular matrix
$$
P = 
\begin{pmatrix}
I_q & -\frac{1}{2(1-\alpha)}C \\
\mathbf{0}_q & I_q
\end{pmatrix},
$$
we obtain
\begin{equation}\label{congruente}
P^TBP = (1-\alpha) 
\begin{pmatrix}
\mathbf{0}_q & I_q \\
I_q & \mathbf{0}_q
\end{pmatrix}.
\end{equation}

By Theorem \ref{sylvesterlaw}, $B$ has exactly $q$ positive and $q$ negative eigenvalues, given that the matrix on the right-hand side of \eqref{congruente} has eigenvalues $(1-\alpha)$ and $-(1-\alpha)$, both with multiplicity $q$.

Since $B$ is a principal submatrix of $A_\alpha(G) - \alpha I$, from Theorem~\ref{interlacing:matrizes} we have
$$\lambda_i(A_{\alpha}(G) - \alpha I) \geq \lambda_i(B) \geq \lambda_{i+(n-2q)}(A_{\alpha}(G) - \alpha I), \quad 1 \leq i \leq 2q,$$
which means that $A_\alpha(G) - \alpha I$ has at least $q$ positive eigenvalues and at least $q$ negative eigenvalues. Consequently, $A_\alpha(G)$ has at least $q$ eigenvalues strictly greater than $\alpha$ and at least $q$ eigenvalues strictly less than $\alpha$. Therefore, applying Corollary \ref{cotasalphaautovalores}, we obtain
$$
m_{A_\alpha(G)}[\psi, \alpha) \ge q(G)
$$
and
$$
m_{A_\alpha(G)}(\alpha, \Delta(G)] \ge q(G),
$$
which concludes the proof.
\end{proof}

\begin{cor}
Let $G$ be a connected graph of order $n \ge 3$. If $\alpha \in [0, 1)$, then
$$
m_{A_\alpha(G)}(\alpha) \le n - 2q(G).
$$
\end{cor}
\begin{proof}
    Note that
    $$m_{A_\alpha(G)}[\psi, \Delta(G)] = m_{A_\alpha(G)}[\psi, \alpha) + m_{A_\alpha(G)}(\alpha) + m_{A_\alpha(G)}(\alpha, \Delta(G)].$$
    From Corollary \ref{cotasalphaautovalores}
    $$m_{A_\alpha(G)}(\alpha) = n -  m_{A_\alpha(G)}[\psi, \alpha) - m_{A_\alpha(G)}(\alpha, \Delta(G)],$$
    and from Theorem \ref{cotaquasi}
    $$m_{A_\alpha(G)}(\alpha) \le n - 2q(G),$$
    which concludes the proof.
\end{proof}

Let $G$ be a connected graph of order $n \ge 3$ and ${R(G) = V(G)\setminus(\mathcal{P}(G) \cup \mathcal{Q}(G)))}$. Thus, $R(G)$ is the set of internal vertices that are not quasi-pendant. The following result characterizes the graphs for which equality holds simultaneously in the two bounds of Theorem \ref{org:dist1} and, separately, in the two bounds of Theorem \ref{cotaquasi}.

\begin{prop}
Let $G$ be a connected graph of order $n \geq 3$.
\begin{enumerate}[(i)]
    \item For $\alpha \in [0,1)$, the following two equalities
    \[
        m_{A_{\alpha}(G)}[\psi,\alpha] = p(G)
        \quad\text{and}\quad
        m_{A_{\alpha}(G)}[\alpha,\Delta(G)] = p(G)
    \]
    hold simultaneously if and only if every internal vertex of $G$ is quasi-pendant.

    \item Consider the following two equalities:
    \[
        m_{A_{\alpha}(G)}[\psi,\alpha) = q(G)
        \quad\text{and}\quad
        m_{A_{\alpha}(G)}(\alpha,\Delta(G)] = q(G).
    \]
    For $\alpha \in (0,1)$, these equalities hold simultaneously if and only if
    $R(G) = \emptyset$. If $\alpha = 0$, they hold simultaneously if and only if
    $R(G)$ is an independent set.
\end{enumerate}
\end{prop}

\begin{proof}
Let $G$ be a connected graph of order $n \ge 3$.  From Corollary \ref{cotasalphaautovalores},   \begin{equation}\label{igualdadealpha2x2}
    m_{A_{\alpha}(G)}[\psi, \alpha] + m_{A_{\alpha}(G)}[\alpha, \Delta(G)] = n + m_{A_{\alpha}(G)}(\alpha)
    \end{equation}
    and
\begin{equation}\label{igualdadealpha2x3}
      m_{A_{\alpha}(G)}[\psi, \alpha) + m_{A_{\alpha}(G)}( \alpha, \Delta(G)] = n - m_{A_{\alpha}(G)}(\alpha),
    \end{equation}
where $n=p(G)+q(G)+|R(G)|.$ From Theorem~\ref{teo:germaincardoso},
$$
m_{A_{\alpha}(G)}(\alpha)
=
p(G)-q(G)+m_{N_{\alpha}(G)}(\alpha).
$$

For item~(i), substituting the above expression into Equation~\eqref{igualdadealpha2x2} yields
$$
\begin{aligned}
&\big(m_{A_{\alpha}(G)}[\psi,\alpha]-p(G)\big)
+
\big(m_{A_{\alpha}(G)}[\alpha,\Delta(G)]-p(G)\big)\\
&\qquad=
|R(G)|+m_{N_{\alpha}(G)}(\alpha).
\end{aligned}
$$

By Theorem~\ref{org:dist1}, both terms on the left-hand side are nonnegative. Since the terms on the right-hand side are also nonnegative, the two equalities
$$
m_{A_{\alpha}(G)}[\psi,\alpha]=p(G)
\quad\text{and}\quad
m_{A_{\alpha}(G)}[\alpha,\Delta(G)]=p(G)
$$
hold simultaneously if and only if
$$
|R(G)|=0
\quad\text{and}\quad
m_{N_{\alpha}(G)}(\alpha)=0.
$$
Since $|R(G)|=0$ implies that $N_{\alpha}(G)$ has order zero, the latter condition is then automatic. Thus, the two equalities hold simultaneously if and only if
$$
R(G)=\emptyset.
$$
By the definition of $R(G)$, this is equivalent to every internal vertex of $G$ being quasi-pendant. Hence, item~(i) follows.

For item~(ii), substituting the expression for $m_{A_{\alpha}(G)}(\alpha)$ into Equation~\eqref{igualdadealpha2x3} gives
$$
\begin{aligned}
&\big(m_{A_{\alpha}(G)}[\psi,\alpha)-q(G)\big)
+
\big(m_{A_{\alpha}(G)}(\alpha,\Delta(G)]-q(G)\big)\\
&\qquad=
|R(G)|-m_{N_{\alpha}(G)}(\alpha).
\end{aligned}
$$
By Theorem~\ref{cotaquasi}, both terms on the left-hand side are nonnegative. Hence, the two equalities
$$
m_{A_{\alpha}(G)}[\psi,\alpha)=q(G)
\quad\text{and}\quad
m_{A_{\alpha}(G)}(\alpha,\Delta(G)]=q(G)
$$
hold simultaneously if and only if
\begin{equation}\label{eq:RNalpha}
m_{N_{\alpha}(G)}(\alpha)=|R(G)|.
\end{equation}
Since $N_{\alpha}(G)$ is a real symmetric matrix of order $|R(G)|$, Equation~\eqref{eq:RNalpha} holds if and only if all eigenvalues of $N_{\alpha}(G)$ are equal to $\alpha$, that is,
$$
N_{\alpha}(G)=\alpha I_{|R(G)|}.
$$
Suppose firstly that $\alpha\in(0,1)$. If $R(G)\neq\emptyset$, then there exists a vertex $v\in R(G)$ with $d(v)\geq2$. Since
$$
A_{\alpha}(G)=\alpha D(G)+(1-\alpha)A(G),
$$
the diagonal entry of $N_{\alpha}(G)$ corresponding to $v$ is
$$
(N_{\alpha}(G))_{vv}=\alpha d(v)\geq2\alpha>\alpha.
$$
Thus,
$$
N_{\alpha}(G)\neq\alpha I_{|R(G)|},
$$
a contradiction. Therefore, $R(G)=\emptyset$. Conversely, if $R(G)=\emptyset$, then $N_{\alpha}(G)$ has order zero and Equation~\eqref{eq:RNalpha} holds trivially. Hence, for $\alpha\in(0,1)$, the two equalities in item~(ii) hold simultaneously if and only if $R(G)=\emptyset.$
Finally, let $\alpha=0$. In this case,
$$
N_0(G)=A(G[R(G)]),
$$
the adjacency matrix of the subgraph induced by $R(G)$. Since $N_0(G)$ is symmetric,
$$
m_{N_0(G)}(0)=|R(G)|
$$
if and only if all eigenvalues of $N_0(G)$ are zero, which is equivalent to
$$
N_0(G)=\mathbf{0}_{|R(G)|}.
$$
Therefore,
$$
m_{N_0(G)}(0)=|R(G)|
$$
if and only if $G[R(G)]$ has no edges, that is, if and only if $R(G)$ is an independent set. This proves item~(ii).
\end{proof}

\section{Distribution of $A_{\alpha}(G)$-eigenvalues involving $\gamma(G), \beta(G)$ and $\nu(G)$}
\label{sec:main2}

In this section, we present results on the distribution of the $A_{\alpha}(G)$-eigenvalues for $\alpha \in [1/2, 1)$, involving domination number, $\gamma(G)$; edge covering number, $\beta(G)$; and matching number, $\nu(G).$ From Proposition \ref{prop:positive} and Lemma \ref{cotaniki}, if $\alpha \in [\frac{1}{2},1]$ then $\lambda_i(A_{\alpha}(G)) \geq 0,$ and  $\lambda_i(A_{\alpha}(G)) \leq \Delta(G),$ for all  $1 \leq i \leq n$. So for $\alpha \in [\frac{1}{2},1],$ $\lambda_i(A_{\alpha}(G)) \in [0,\Delta(G)].$ As $A_{1/2}(G) = \frac{1}{2}Q(G)$, the results presented in this section are valid for the signless Laplacian matrix. 

\subsection{Distribution involving $\gamma(G)$}

In \cite{HedetniemiJacobsTrevisan2016}, the authors presented results for the distribution of $L(G)$-eigenvalues involving the graph and its subgraphs obtained by edge removal. The next proposition establishes a relation between the number $A_{\alpha}(G)$-eigenvalues 
in the intervals $[0, a)$ and $[a, \Delta]$, where $a \in \mathbb{R}$, with the number of  $A_{\alpha}(H)$-eigenvalues when  $H$ is a subgraph of $G$ obtained by removing some edges. 

\begin{prop}\label{lem:domtreuil}
Let $G = (V, E)$ and $H = (V, F)$ be graphs of order $n$ such that $F \subseteq E$. If $a \in [0, \Delta(G)]$ and $\alpha \in [1/2, 1]$, then
\begin{enumerate}[(i)]
\item $m_{A_{\alpha}(H)}[0, a) \geq m_{A_{\alpha}(G)}[0, a);$\label{lem:domtreuilitem1}
\item $m_{A_{\alpha}(G)}[a, \Delta(G)] \geq m_{A_{\alpha}(H)}[a, \Delta(H)].$\label{lem:domtreuilitem2}
\end{enumerate}
\end{prop}
\begin{proof}
Let $a \in [0, \Delta(G)]$ and $\alpha \in [1/2,1]$. Suppose, by contradiction, that $m_{A_{\alpha}(H)}[0,a) < m_{A_{\alpha}(G)}[0,a)$. Since $G$ and $H$ are of order $n$, if $m_{A_{\alpha}(H)}[0,a) < m_{A_{\alpha}(G)}[0,a)$, there exists at least one $i$, $1 \leq i \leq n$ such that $\lambda_i(A_{\alpha}(H)) > \lambda_i(A_{\alpha}(G))$, which is a contradiction by Proposition~\ref{prop:linxue}. Thus, item~\eqref{lem:domtreuilitem1} follows. Similarly, item~\eqref{lem:domtreuilitem2} is proven.
\end{proof}

As every graph has a spanning forest, from Proposition~\ref{lem:domtreuil}, we obtain the following corollary.

\begin{cor}\label{cor:domtreuil}
Let $G$ be a graph and $T$ a spanning forest of $G$. If $a \in [0, \Delta(G)]$ and $\alpha \in [1/2, 1]$, then
\begin{enumerate}[(i)]
\item $m_{A_{\alpha}(T)}[0, a) \geq m_{A_{\alpha}(G)}[0, a);$\label{cor:domtreuilitem1}
\item $m_{A_{\alpha}(G)}[a, \Delta(G)] \geq m_{A_{\alpha}(T)}[a, \Delta(T)].$\label{cor:domtreuilitem2}
\end{enumerate}
\end{cor}

The next corollary presents bound for the number of $A_{\alpha}(G)$-eigenvalues in the intervals $[\alpha, \Delta(G)]$ and $\left[\dfrac{3\alpha + \sqrt{9\alpha^2-16\alpha+8}}{2}, \Delta(G)\right]$ based on the order of the graph.
\begin{cor}
    Let $G$ be a graph of order $n$ that has a Hamiltonian path. If $\alpha \in [1/2, 1]$ then $$m_{A_{\alpha}(G)}[\alpha, \Delta(G)] \geq 2\left\lfloor\dfrac{n}{3}\right\rfloor \text{ and } m_{A_{\alpha}(G)}\left[\dfrac{3\alpha + \sqrt{9\alpha^2-16\alpha+8}}{2}, \Delta(G)\right] \geq \left\lfloor\dfrac{n}{3}\right\rfloor.$$ 
\end{cor}
\begin{proof}
Let $G$ be a graph of order $n$, and let $\alpha \in [1/2, 1]$. Suppose that $G$ has a Hamiltonian path and choose a spanning forest\[F\cong \left\lfloor\frac n3\right\rfloor S_3\;\cup\;\left(n-3\left\lfloor\frac n3\right\rfloor\right)K_1,\]using edges of the Hamiltonian path. From Proposition~\ref{starspectrum}, each star $S_3$  has at least two $A_{\alpha}(S_3)$-eigenvalues in the interval $[\alpha, 2]$ and therefore, the sum of the multiplicities of the $A_{\alpha}(F)$-eigenvalues in the interval $[\alpha, 2]$ is equal to $2 \left\lfloor \dfrac{n}{3} \right\rfloor.$
Moreover, from Corollary~\ref{cor:domtreuil}, we have
$$m_{A_{\alpha}(G)}[\alpha, \Delta(G)] \geq m_{A_{\alpha}(F)}[\alpha, 2],$$ and thus
    $$
    m_{A_{\alpha}(G)}[\alpha, \Delta(G)] \geq 2 \left\lfloor \dfrac{n}{3} \right\rfloor.
    $$
    Furthermore, each star $S_3$ also has an $A_{\alpha}(S_3)$-eigenvalue in the interval $[\lambda_1(A_{\alpha}(S_3)), 2]$, where $\lambda_1(A_{\alpha}(S_3)) = \dfrac{3\alpha + \sqrt{9\alpha^2-16\alpha+8}}{2}$. Therefore,
    $$
    m_{A_{\alpha}(G)}\left[\dfrac{3\alpha + \sqrt{9\alpha^2-16\alpha+8}}{2}, \Delta(G)\right] \geq \left\lfloor\dfrac{n}{3}\right\rfloor,
    $$
    and the result follows.
\end{proof}

For a graph $G$ and $\alpha \in [1/2, 1]$, the next theorem establishes a relation between the number of $A_{\alpha}(G)$-eigenvalues in the interval $[0, \alpha)$ and the domination number of $G$, $\gamma(G)$.

\begin{teo}\label{teo:hedtreuil}
If $G$ is a graph and $\alpha \in [1/2,1]$, then $$m_{A_{\alpha}(G)}[0, \alpha) \leq \gamma(G).$$
If $\alpha\in[1/2,1)$ and $G\cong S_n$, then equality holds.
\end{teo}
\begin{proof}
We first suppose that $G$ has no isolated vertices. From Lemma~\ref{domtrelem6}, $G$ has a spanning forest $S_{n_1} \cup S_{n_2} \cup \dots \cup S_{n_{\gamma}(G)}$ disjoint by vertices. Let $F \cong \bigcup_{i=1}^{\gamma(G)} S_{n_i}.$ 
From Proposition~\ref{lem:domtreuil}, it follows that  
$$m_{A_{\alpha}(F)}[0,\alpha) \geq m_{A_{\alpha}(G)}[0,\alpha).$$  
As each star $S_{n_i}$ is a connected component of $F$, from Proposition~\ref{starspectrum}, $A_{\alpha}(S_{n_i})$ has exactly one eigenvalue in the interval $[0,\alpha)$. Then, the sum of the $A_{\alpha}(F)$-eigenvalues in the interval $[0, \alpha)$ is given by
$$m_{A_{\alpha}(F)}[0,\alpha) = \sum_{i=1}^{\gamma(G)} m_{A_{\alpha}(S_{n_i})}[0, \alpha) = \gamma(G).$$ 
Now suppose that $G$ has $k$ isolated vertices, consider $G'$ a subgraph of $G$ that contains no isolated vertices, whose domination number is $\gamma(G^{'}).$ It is easy to see that $m_{A_{\alpha}(G)}[0,\alpha) = m_{A_{\alpha}(G')}[0,\alpha) + k$ and $\gamma(G) = \gamma(G^{'}) + k.$ So, the result follows.  
If $G$ is the star, from Proposition~\ref{starspectrum} the equality is achieved.
\end{proof}

As $m_{A_{\alpha}(G)}[\alpha,\Delta(G)] = n - m_{A_{\alpha}(G)}[0,\alpha)$,  from Theorem~\ref{teo:hedtreuil} we obtain the following remark.

\begin{obs}\label{cor:teo:hedtreuil}
Let $G$ be a graph of order $n$ and $\alpha \in [1/2,1]$. Then $$m_{A_{\alpha}(G)}[\alpha,\Delta(G)] \geq n - \gamma(G).$$
\end{obs}

The next result characterizes the equality $q(G)=\gamma(G)$ in terms of the structure of $G$ and describes its consequences for the distribution of the $A_{\alpha}(G)$-spectrum.

\begin{prop}\label{prop:q_gamma}
Let $G$ be a connected graph of order $n \ge 3$. Then
$$q(G) \le \gamma(G).$$
Moreover, equality holds if and only if $\mathcal{Q}(G)$ is a dominating set of $G$. In this case, for any $\alpha \in [1/2, 1)$
$$m_{A_{\alpha}(G)}[0, \alpha) = q(G) = \gamma(G) \quad \text{and} \quad m_{A_{\alpha}(G)}[\alpha, \Delta(G)] = n - \gamma(G).$$
\end{prop}
\begin{proof}
Let $\mathcal{Q}(G) = \{v_1, \dots, v_q\}$. For each $i$, $1 \leq i \leq q$, choose a pendant vertex $u_i$ adjacent to $v_i$. Since every pendant vertex has exactly one neighbor, the vertices $u_1, \dots, u_q$ are pairwise distinct. Consequently, the edges $\{u_i, v_i\}$ for $1 \le i \le q$ are pairwise disjoint.

Let $U$ be any dominating set of $G$. Since $u_i$ is a pendant vertex whose only neighbor is $v_i$, in order to dominate $u_i$, the set $U$ must contain at least one of $u_i$ or $v_i$. As the pairs ${u_i,v_i}$ are pairwise disjoint, it follows that
$
|U|\geq q.
$ Taking $U$ as a minimum dominating set, we obtain $\gamma(G) \ge q(G)$.

Suppose now that $\gamma(G)=q(G)$, and let $U$ be a minimum dominating set. Then $|U|=q$, and therefore $U$ contains exactly one vertex from each pair ${u_i,v_i}$, $1\leq i\leq q$.

If $u_i\in U$, replace $u_i$ by $v_i$. This replacement preserves domination, since $u_i$ is a pendant vertex and $v_i$ is its only neighbor. In particular, $u_i$ remains dominated by $v_i$, while no other vertex loses its only dominator as a consequence of this replacement. Repeating this procedure for every pendant vertex belonging to $U$, we obtain a dominating set $U'\subseteq\mathcal{Q}(G)$ with
$
|U'|=|U|=q.
$
Since $|\mathcal{Q}(G)|=q$, we must have
$
U'=\mathcal{Q}(G).
$
Hence, $\mathcal{Q}(G)$ is a dominating set of $G$. Conversely, if $\mathcal{Q}(G)$ is a dominating set, then $\gamma(G) \le |\mathcal{Q}(G)| = q(G)$, which implies $\gamma(G) = q(G)$.

Now, assume $\alpha \in [1/2, 1),$ and consider that $\gamma(G) = q(G)$. By Theorem \ref{cotaquasi}, $A_{\alpha}(G)$ has at least $q(G)$ eigenvalues strictly smaller than $\alpha$. Since $A_{\alpha}(G)$ is positive semidefinite for $\alpha \ge 1/2$, these eigenvalues all belong to $[0, \alpha)$. Thus, $q(G) \le m_{A_{\alpha}(G)}[0, \alpha)$. From Theorem \ref{teo:hedtreuil}, we have the upper bound $m_{A_{\alpha}(G)}[0, \alpha) \le \gamma(G)$. Then $m_{A_{\alpha}(G)}[0, \alpha) = \gamma(G)$. Finally, taking the complement in the spectrum, we conclude $m_{A_{\alpha}(G)}[\alpha, \Delta(G)] = n - m_{A_{\alpha}(G)}[0, \alpha) = n - \gamma(G)$.
\end{proof}

\subsection{Distribution involving $\beta(G)$ and $\nu(G)$}

The following lemma provides a structural result on matchings in trees that is fundamental to the proof of Theorem \ref{teo:fanguil}.

\begin{lem}\label{lem:nonpendantedge}
    Let $T$ be a tree of order $n\ge 4$ which is not a star. Then there exists a non-pendant edge $e\in E(T)$ such that
    \[\beta(T-e)=\beta(T).\]
\end{lem}
\begin{proof}
Let $T'$ be the subgraph of $T$ induced by its internal vertices. Since $T$ is not a star and $n\geq 4$, the graph $T'$ is a tree with at least one edge. We consider two cases according to the number of edges in $T'$.

\textit{Case 1:} Suppose that $T'$ has exactly one edge, say $e=\{u,v\}$. Then $T$ is a double star. Choose one pendant edge incident with $u$ and one pendant edge incident with $v$. These two edges form a matching of size $2$. Since every edge of a double star is incident with one of its two internal vertices, no matching can contain more than two edges. Hence,
$\nu(T)=2,$ and there exists a maximum matching $M$ of $T$ that does not contain the non-pendant edge $e$. Therefore, $M$ is also a matching in $T-e$, and so
$$
\nu(T-e)\geq |M|=\nu(T).
$$
Since deleting an edge cannot increase the matching number, we obtain
$$\nu(T-e)=\nu(T). $$   

\textit{Case 2:} Suppose that $T'$ has at least two edges. Since $T'$ is a tree, it contains two adjacent edges, say $e_1$ and $e_2$. Let $M$ be a maximum matching of $T$. Since $M$ is a matching, it cannot contain both $e_1$ and $e_2$. Thus, at least one of these edges, say $e_1$, does not belong to $M$. Since $e_1$ is a non-pendant edge of $T$, the matching $M$ is also a matching of $T-e_1$. Consequently,
$$
\nu(T-e_1)\geq |M|=\nu(T).
$$
Again, since deleting an edge cannot increase the matching number,
$$
\nu(T-e_1)=\nu(T).
$$
Thus, in both cases, there exists a non-pendant edge $e\in E(T)$ such that $\nu(T-e)=\nu(T)$. Since $e$ is non-pendant, the forest $T-e$ has no isolated vertices. Therefore, by Lemma~\ref{Galli1959},
$$\beta(T-e)=n-\nu(T-e)=n-\nu(T)=\beta(T),$$  
and the result follows.
\end{proof}

The following theorem shows that the edge covering number of $G$, $\beta(G)$, provides a lower bound for the number of $A_{\alpha}(G)$-eigenvalues in the interval $[\alpha,\Delta(G)]$ and gives a sufficient condition for equality in terms of the matching number $\nu(G)$ and $q(G)$.

\begin{teo}\label{teo:fanguil}
    If $G$ is a connected graph of order $n \geq 2$, and $\alpha \in [1/2,1]$, then 
    $$m_{A_{\alpha}(G)}[\alpha, \Delta(G)] \geq \beta(G).$$
    If $\nu(G)=q(G)$ and $\alpha\in[1/2,1)$, then equality holds.
\end{teo}
\begin{proof}
    Let $G$ be a connected graph and $\alpha \in [1/2, 1]$. It is easy to see that $\beta(G) \leq \beta(T)$ for any spanning tree $T$ of $G$. So, it is enough to prove the result for $T$. 
    
    The proof is by induction on the order of $T$. For $n = 2$ the result follows from Proposition~\ref{starspectrum}.  
    Assume inductively that the assertion holds for every tree of order at most $k$. Let $T$ be a spanning tree of $G$ of order $k+1$. If $T \cong S_{k+1}$, from Proposition~\ref{starspectrum} we have
    $$
    m_{A_{\alpha}(T)}[\alpha, \Delta(T)] = k = \beta(T),
    $$  
    and the result follows. 
    
    Now suppose $T \ncong S_{k+1}$. By Lemma~\ref{lem:nonpendantedge}, there exists a non-pendant edge $e \in E(T)$ such that $\beta(T-e) = \beta(T)$. Since $e$ is a non-pendant edge, its removal results in two connected components, $T_1$ and $T_2$, both of which are trees of order at least two (hence without isolated vertices). Applying the induction hypothesis and from Proposition~\ref{lem:domtreuil}, we have  
    $$
    \begin{aligned}
    m_{A_{\alpha}(T)}[\alpha, \Delta(T)] &\geq m_{A_{\alpha}(T-e)}[\alpha, \Delta(T-e)] \\
    &= m_{A_{\alpha}(T_1)}[\alpha, \Delta(T_1)] + m_{A_{\alpha}(T_2)}[\alpha, \Delta(T_2)] \\
    &\geq \beta(T_1) + \beta(T_2) \\
    &= \beta(T-e) \\
    &= \beta(T),
    \end{aligned}
    $$  
    and the result follows. 
    
    Now suppose $\nu(G) = q(G)$. From Lemma~\ref{Galli1959}, we have $n = \beta(G) + \nu(G)$, which implies that  
    $$
    \beta(G) + q(G) = n.
    $$  
    From Theorem~\ref{cotaquasi}, we have that $m_{A_{\alpha}(G)}[0,\alpha) \geq q(G)$. Then  
    $$
    \begin{aligned}
    n &= \beta(G) + \nu(G) \\
    &= \beta(G) + q(G) \\
    &\leq m_{A_{\alpha}(G)}[\alpha, \Delta(G)] + m_{A_{\alpha}(G)}[0,\alpha) = n.
    \end{aligned}
    $$  
    Therefore, $\beta(G) = m_{A_{\alpha}(G)}[\alpha, \Delta(G)]$ and $q(G) = m_{A_{\alpha}(G)}[0,\alpha)$, which concludes the proof. 
\end{proof}

For a connected graph of order $n\geq 3$, every edge cover must contain all pendant edges, since each pendant vertex is incident with a unique edge. Hence,
\begin{equation*}
\beta(G)\geq p(G).    
\end{equation*}
Therefore, the bound in Theorem~\ref{teo:fanguil} improves upon the bound in Theorem~\ref{org:dist1} for $\alpha\in[1/2,1)$.

The following proposition characterizes the exact structural condition for the equality $q(G) = \nu(G)$, identifying when the spectral bounds in terms of the matching number are attained.
\begin{prop}\label{prop:q_nu}
If $G$ is a connected graph of order $n \ge 3$, then
$$q(G) \le \nu(G).$$
Moreover, $q(G) = \nu(G)$ if and only if $R(G)$ is an independent set. In this case, for any $\alpha \in [1/2, 1)$, we have 
$$m_{A_{\alpha}(G)}[0, \alpha) = q(G) = \nu(G) \quad \text{and} \quad m_{A_{\alpha}(G)}[\alpha, \Delta(G)] = \beta(G).$$
\end{prop}
\begin{proof}
Let $\mathcal{Q}(G) = \{v_1, \dots, v_q\}$. For each $v_i \in \mathcal{Q}(G)$, choose a pendant neighbor $u_i$. As established previously, the chosen pendant vertices are pairwise distinct, and therefore 
$$M_0 = \{\{u_iv_i\} : 1 \le i \le q\}$$
is a matching of $G$ with size $q$. Hence, $\nu(G) \ge q(G)$.

Suppose that $\nu(G)=q(G)$. Suppose, to the contrary, that $R(G)$ is not an independent set. Then there exist vertices $x,y\in R(G)$ such that $\{x,y\}\in E(G)$. Since $x$ and $y$ are not incident with any edge of $M_0$, the set \[ M_0\cup\{\{x,y\}\} \] is a matching of $G$ of size $q+1$, which contradicts $\nu(G)=q(G)$. Therefore, $R(G)$ is an independent set.

Conversely, suppose that $R(G)$ is an independent set. We claim that $\mathcal{Q}(G)$ is a vertex cover of $G$. Indeed, every edge incident with a pendant vertex has its other endpoint in $\mathcal{Q}(G)$ by definition. Furthermore, there are no edges with both endpoints in $R(G)$, since $R(G)$ is an independent set. Hence, every edge of $G$ has at least one endpoint in $\mathcal{Q}(G)$. Since every edge of a matching contains a vertex of the vertex cover $\mathcal{Q}(G)$, and distinct edges of the matching contain distinct such vertices, every matching has at most $|\mathcal{Q}(G)| = q$ edges. Therefore, $\nu(G) \le q(G)$. Together with $q(G) \le \nu(G)$, this yields $ q(G)=\nu(G)$.

Now, let $\alpha \in [1/2, 1)$ and suppose that $R(G)$ is an independent set. Then, $q(G) = \nu(G)$. By Theorem \ref{cotaquasi} and the positive semidefiniteness of $A_{\alpha}(G)$ for $\alpha\geq1/2$,
$$m_{A_{\alpha}(G)}[0, \alpha) \ge q(G) = \nu(G).$$
On the other hand, from Theorem \ref{teo:fanguil} and Lemma \ref{Galli1959} we have,
$$m_{A_{\alpha}(G)}[0, \alpha) = n - m_{A_{\alpha}(G)}[\alpha, \Delta(G)] \le n - \beta(G) = \nu(G).$$
Thus, $m_{A_{\alpha}(G)}[0, \alpha) = q(G) = \nu(G)$. Taking complements in the spectrum yields
$$m_{A_{\alpha}(G)}[\alpha, \Delta(G)] = n - \nu(G) = \beta(G).$$

\end{proof}

The following corollary combines the previous bounds into a single result, establishing the exact structural conditions under which all equalities are attained simultaneously.
\begin{cor}\label{cor:sandwich_bounds}
Let $G$ be a connected graph of order $n \ge 3$ and let $\alpha \in [1/2, 1)$. Then
\begin{equation}\label{eq:sandwich1}
q(G) \le m_{A_{\alpha}(G)}[0, \alpha) \le \min\{\gamma(G), \nu(G)\}.
\end{equation}
Equivalently,
\begin{equation}\label{eq:sandwich2}
\max\{n - \gamma(G), \beta(G)\} \le m_{A_{\alpha}(G)}[\alpha, \Delta(G)] \le n - q(G).
\end{equation}
Moreover,
\begin{enumerate}[(i)]
    \item if $\mathcal{Q}(G)$ is a dominating set, then $m_{A_{\alpha}(G)}[0, \alpha) = q(G) = \gamma(G)$ and $m_{A_{\alpha}(G)}[\alpha, \Delta(G)] = n - \gamma(G)$;\label{cor:sandwich_boundsi}
    \item if $R(G)$ is an independent set, then $m_{A_{\alpha}(G)}[0, \alpha) = q(G) = \nu(G)$ and $m_{A_{\alpha}(G)}[\alpha, \Delta(G)] = \beta(G)$;\label{cor:sandwich_boundsii}
    \item if $\mathcal{Q}(G)$ is a dominating set and $R(G)$ is an independent set, then $q(G) = \gamma(G) = \nu(G)$, and all bounds in \eqref{eq:sandwich1} and \eqref{eq:sandwich2} are attained simultaneously.\label{cor:sandwich_boundsiii}
\end{enumerate}
\end{cor}
\begin{proof}
By Theorem \ref{cotaquasi}, $A_{\alpha}(G)$ has at least $q(G)$ eigenvalues strictly smaller than $\alpha$. Since $A_{\alpha}(G)$ is positive semidefinite for $\alpha \in [1/2, 1)$, all these eigenvalues lie in $[0, \alpha)$. Hence, 
$$q(G) \le m_{A_{\alpha}(G)}[0, \alpha).$$
Theorem \ref{teo:hedtreuil} yields
$$m_{A_{\alpha}(G)}[0, \alpha) \le \gamma(G).$$
On the other hand, from Theorem \ref{teo:fanguil} we have
$$m_{A_{\alpha}(G)}[\alpha, \Delta(G)] \ge \beta(G).$$
Since $m_{A_{\alpha}(G)}[0, \alpha) + m_{A_{\alpha}(G)}[\alpha, \Delta(G)] = n$, we obtain
$$m_{A_{\alpha}(G)}[0, \alpha) \le n - \beta(G) = \nu(G).$$
So, this proves the inequalities \eqref{eq:sandwich1}. Subtracting the inequalities in \eqref{eq:sandwich1} from $n$ and using $n - \nu(G) = \beta(G),$ we obtain the inequalities \eqref{eq:sandwich2}.

Statement \eqref{cor:sandwich_boundsi} follows from Proposition \ref{prop:q_gamma}, statement \eqref{cor:sandwich_boundsii} follows from Proposition \ref{prop:q_nu}, while statement \eqref{cor:sandwich_boundsiii} follows by combining the two.
\end{proof}

Classical results in spectral graph theory connect the matching number to the nullity of the adjacency matrix. In particular, for trees, $\eta(T)=n-2\nu(T),$
where $\eta(T)$ denotes the nullity of $T$ \cite{Cvetkovic1980}. In this context, Corollary~\ref{cor:sandwich_bounds} extends this classical connection to the $A_{\alpha}(G)$-matrices. Rather than relating the matching number solely to the nullity, it shows that $\nu(G)$ provides an upper bound for the number of eigenvalues of $A_{\alpha}(G)$ in the interval $[0,\alpha)$, for $\alpha\in[1/2,1)$, with equality characterized by Proposition~\ref{prop:q_nu}.

\section*{Acknowledgments} 

\textbf{Funding information:} 
The authors thank CEFET/RJ – Centro Federal de Educação Tecnológica Celso Suckow da Fonseca for institutional support. They also acknowledge partial financial support from the Fundação Carlos Chagas Filho de Amparo à Pesquisa do Estado do Rio de Janeiro (FAPERJ), Process SEI 260003/001228/2023 and Grant E-20/2022-284573, and from Conselho Nacional de Desenvolvimento Científico e Tecnológico (CNPq), Grant 405552/2023-8.

\subsection*{Conflict of interest}

Authors state no conflict of interest.

\section*{Data availability}
		
Data sharing not applicable to this article as no datasets were generated or analyzed during the current study.

\bibliographystyle{plain} 
\bibliography{reference}

\end{document}